\documentclass[11pt,fleqn]{article}

\usepackage{amsmath,amssymb,amsthm,cite,enumerate,graphics,epsfig}%

\newcommand{\todo}[1][\null]{\ensuremath{\clubsuit}}
\newcommand{\lsemioplus}{\mathbin{\mbox{$\lefteqn{\hspace{.77ex}\rule{.4pt}{1.2ex}}{\in}$}}}
\newcounter{tbn}

\newcounter{mcasenum}

\newtheorem{theorem}{Theorem}
\newtheorem{lemma}{Lemma}

\newtheorem*{proposition*}{Proposition}
{\theoremstyle{definition}

}

\begin{document}

\par\noindent {\LARGE\bf
 Lie symmetries of generalized Burgers equations:\\ application to  boundary-value problems
\par}

{\vspace{3mm}\par\noindent\large O.\,O.~Vaneeva$^{\dag 1}$, C. Sophocleous$^{\ddag 2}$ and P.\,G.\,L. Leach$^{\ddag\S 3}$
\par\vspace{5mm}\par}

{\par\noindent\it
${}^\dag$\ Institute of Mathematics of NAS of Ukraine, 3 Tereshchenkivska Str., Kyiv-4, 01601 Ukraine\\[1ex]
${}^\ddag$\ Department of Mathematics and Statistics, University of Cyprus, Nicosia CY 1678, Cyprus\\[1ex]
${}^\S$\ School of Mathematical Sciences, University of KwaZulu-Natal, Private Bag X54001,\\
$\phantom{{}^\S}$\ Durban 4000, South Africa
}
\vspace{3mm}
{\par\noindent
$\phantom{{}^\dag{}\;}\ $E-mail: \it$^1$vaneeva@imath.kiev.ua,
$^2$christod@ucy.ac.cy,
$^3$leach.peter@ucy.ac.cy
\par}

{\vspace{5mm}\par\noindent\hspace*{8mm}\parbox{146mm}{\small
 There exist several approaches exploiting Lie symmetries in the reduction of  boundary-value problems for partial differential equations modelling real-world phenomena to those problems for ordinary differential equations. Using an example of generalized Burgers equations appearing in nonlinear acoustics we show that that the ``direct'' procedure of  solving boundary-value problems using Lie symmetries firstly described by Bluman is more general and straightforward than the method suggested by Moran and Gaggioli in [J Eng Math {\bf 3} (1969), 151--162].
 After the group classification of a class of generalized Burgers equations with time-dependent viscosity is performed we solve an associated
 boundary-value problem using the symmetries obtained.
}\par\vspace{3mm}}


\section{Introduction}

Nonlinear partial differential equations (PDEs) and their systems  often appear in physical sciences and engineering as models describing real-world phenomena.
In applications one is usually interested not in every solution
of a PDE, but in a solution satisfying some additional conditions such as an
initial condition  and/or a boundary condition.

 Lie symmetry methods play an important role in solving nonlinear PDEs providing us with the algorithmic method of Lie reduction. There exist several approaches exploiting Lie symmetries in reduction of boundary-value problems (BVPs) for PDEs to those for ODEs.
The classical technique
is to require that both equation and boundary conditions are left invariant under the action of a one-parameter Lie group
of transformations. Of course the infinitesimal approach is usually applied, i.e., a basis of operators of Lie invariance algebra is used instead of finite transformations from the corresponding Lie symmetry group  (see, e.g.,~\cite[Section 4.4]{Bluman&Anco2002}).
The first works in this direction appeared in the late sixties (see, e.g.,~\cite{Moran&Gaggioli1968,Moran&Gaggioli1969,Bluman&Cole1969,Bluman1974}).
Bluman used the  approach~\cite{Bluman&Cole1969,Bluman1974}, that  generally can be termed as the  ``direct'' one, namely firstly the symmetries of a PDE were derived and then the boundary conditions were checked to determine whether they are also invariant under the action of the generators of symmetry found. In the case of a positive answer the BVP for the PDE was reduced to a BVP for an ODE. Using this technique a number of boundary-value problems were solved (see, e.g.,~\cite{Sophocleous&OHara&Leach2011a,Sophocleous&OHara&Leach2011b,Burde1996,Sophocleous&OHara&Leach2013}).

The method suggested by Moran and Gaggioli in~\cite{Moran&Gaggioli1969} uses specific  one-parameter Lie groups of transformations of the independent and dependent variables of the PDE system as well as of all arbitrary elements which appear in the equations under study and in initial and boundary conditions.
Namely, only the groups of scalings and translations are considered which can lead to self-similar or travelling-wave solutions only. After the admitted Lie group of scalings and/or translations is specified,  the complete set of absolute invariants has to be found. Then a  boundary-value problem for the PDE system is reduced to similar but simpler problem for the ODE system. Such an approach was applied to a number of engineering  problems (see, e.g.,~\cite{AbdelMalek&Helal2011} and references therein).

 There is also the approach in which the group classification of a PDE system and associated  boundary conditions is performed simultaneously (see, e.g.,~\cite{Kovalenko,Kovalenko_dis}). The usage of symmetries in the course of
invariant parameterization or discretization of the system under
consideration is discussed in~\cite{BiPo2012,PoBi2012}. Lie symmetries can also be used for certain cases when boundary conditions themselves are not invariant
with respect to the corresponding Lie group of transformations~\cite{Goard14}.

In this paper we demonstrate that the ``direct'' approach is much easier than that one suggested in~\cite{Moran&Gaggioli1969} and  used, e.g., in~\cite{AbdelMalek&Helal2011}, especially taking into account that problems of group classification  are solved already for wide classes of nonlinear PDEs. To illustrate this we use an example of a generalized Burgers equation of the form
\begin{equation}\label{Eq_GenBurgers}
u_t+a(u^n)_x=g(t)u_{xx},
\end{equation}
where $a$ is a nonzero constant, $g$ is an arbitrary smooth nonvanishing function of $t$ and $n\neq0,1$.

If $n=2$, $a=1/2$, and $g=-\nu$, where $\nu$ is a nonzero constant, equation~\eqref{Eq_GenBurgers} becomes
the prominent Burgers equation,
$
u_t + uu_x + \nu u_{xx} = 0,
$
that is one of the simplest nonlinear $(1 + 1)$ evolution equations that is exactly solvable.  It has a long history as it was already known to Forsyth \cite {Forsyth06} and discussed by Bateman not many years later \cite{Bateman15}.  However, it was a serious contribution made by Burgers which led to its present name \cite{Burgers48}.  Burgers equation has been used to describe many processes in fluid mechanics and a variety of other fields which seem to be rather disparate.  Its remarkable feature is that it can be transformed to the standard heat equation by means of the Hopf--Cole transformation \cite{Hopf50,Cole51}. Therefore it is $C$-integrable~\cite{Calogero1991}.

The generalized Burgers equations~\eqref{Eq_GenBurgers} with $n=2$ and a nonconstant function $g$ were derived in~\cite{HammertonCrighton1989} and describe
 the propagation of weakly nonlinear acoustic waves under
the influence of geometrical spreading and thermoviscous diffusion. Lie symmetries of such equations were studied in~\cite{Doyle&Englefield1990,wafo2004d}. This and other generalizations of the Burgers equations  are discussed, e.g., in~\cite{Sachdev1987,Sachdev2000}.
Recently, a quarter-plane problem for the modified Burgers equations $u_t+u^pu_x=u_{xx}$, $p>1$, was investigated in~\cite{JALeach2013}.

We solve the group classification problem for class~\eqref{Eq_GenBurgers} in the framework of modern group analysis.
This problem  is formulated as follows~\cite{Ovsiannikov1982}: given a class of differential equations, the problem is
to classify all possible cases of extension of Lie invariance algebras of such equations with respect to
the equivalence group of the class.   Then we consider the class of BVPs and solve successfully a specific BVP satisfying a requirement of invariance with respect to Lie symmetries obtained.
In contrast to the work performed in~\cite{AbdelMalek&Helal2011}
we use the ``direct'' approach~\cite{Bluman&Cole1969,Bluman&Anco2002} to solve the equation with associated boundary conditions and show that it is easier to implement and  more transparent.

\section{Equivalence transformations}If two PDEs are connected by a point  transformation, then these equations are called {\it similar}~\cite{Ovsiannikov1982} (it is possible also to consider a similarity up to  contact transformations). Similar PDEs have similar sets of solutions, symmetries, conservation laws and other related information.  Oftentimes equations related in this mathematical way have applications in apparently very distinct areas.  As a simple example, the heat equation found in engineering and the Black-Scholes equation of financial mathematics are equations connected by a point transformation and invariant under the actions of similar representations of the same Lie group \cite{Gasizov98}.
Therefore it is rather important to study point (or even contact) transformations that link equations from a given class of PDEs. Such transformations are called
admissible~\cite{Popovych&Kunzinger&Eshraghi2010} (or form-preserving~\cite{Kingston&Sophocleous1998}, or allowed~\cite{Winternitz92}) ones.
We note that such transformations can be used successfully in the study of integrability~\cite{vane2014b}.

Admissible transformations that
preserve the differential structure of the class and transform only its arbitrary elements
 are called equivalence transformations and form a group.
Notions of different kinds of equivalence group can be found, e.g., in~\cite{vane2012b}.
We were able to study all admissible  transformations in class~\eqref{Eq_GenBurgers}. They appeared to  be exhausted by equivalence once. The results of the study are given in the following statements.
\begin{theorem}
The usual equivalence group~$G^{\sim}$ of class~\eqref{Eq_GenBurgers} comprises the transformations
\[
\begin{array}{l}
\tilde t=\delta_1t+\delta_2,\quad \tilde x=\delta_3x+\delta_4,\quad
\tilde u=\delta_5u, \quad
\tilde a=\dfrac{\delta_3}{\delta_1}\delta_5^{1-n}a, \quad
\tilde g=\dfrac{{\delta_3}^2}{\delta_1} g,\quad \tilde n=n,
\end{array}
\]
where
 $\delta_j,$ $j=1,\dots,5$, are arbitrary constants
with $\delta_1\delta_3\delta_5\not=0$.
\end{theorem}
If $n=2$, class~\eqref{Eq_GenBurgers} admits a nontrivial conditional equivalence group which is wider than~$G^{\sim}$.

\begin{theorem}
The generalized equivalence group~$\hat
G^{\sim}_{2}$ of the class,
\begin{equation}\label{Eq_GenBurgers_n2}
u_t+a(u^2)_x=g(t)u_{xx},
\end{equation}
 consists of  the transformations
\[
\begin{array}{l}
\tilde t=\dfrac{\alpha t+\beta}{\gamma t+\delta},\quad \tilde x=\dfrac{\kappa x +\mu_1t+\mu_0}{\gamma t+\delta},\quad
\tilde u=\dfrac{\sigma}{2a(\alpha\delta-\beta\gamma)}\left(2a\kappa(\gamma t+\delta)u-\kappa\gamma x+\mu_1\delta-\mu_0\gamma\right), \\[2ex]
\tilde a=\dfrac{a}{\sigma} \quad \mbox{\rm and} \quad
\tilde g=\dfrac{\kappa^2}{\alpha\delta-\beta\gamma} g,
\end{array}
\]
where
 $\alpha, \beta, \gamma, \delta, \kappa, \mu_1, \mu_0, \sigma$ are constants defined up to a nonzero multiplier,
 $\alpha\delta-\beta\gamma\neq0$ and $\kappa\sigma\not=0$.
\end{theorem}
\begin{theorem}
Let two equations from class~\eqref{Eq_GenBurgers}, $u_t+a(u^n)_x =g(t)u_{xx}$ and\,
$\tilde u_{\tilde t}+\tilde a(\tilde u^{\tilde n})_{\tilde x} =\tilde g(\tilde t)\tilde u_{\tilde x\tilde x}$,
be connected by a point transformation $\mathcal{T}$ in the variables~$t$, $x$ and~$u$.
Then
the transformation $\mathcal{T}$ is the projection  on the space $(t,x,u)$ of  a transformation from the group
$G^{\sim}$ if $n\not=2$, or from the group  $\hat G^{\sim}_{2}$, if $n=2$.
\end{theorem}
\begin{proof}Suppose that an equation from class~\eqref{Eq_GenBurgers} is connected with an  equation
\begin{equation}\label{Eq_GenBurgers_tilde}
{\tilde u}_{\tilde t}+\tilde a(\tilde u^{\tilde n})_{\tilde x}=\tilde g(\tilde t){\tilde u}_{\tilde x\tilde x}
\end{equation}
 from the same class by a point transformation
$
\tilde t=T(t,x,u),$ $\tilde x=X(t,x,u),$ $\tilde u=U(t,x,u),
$
where $\left|\partial(T,X,U)/\partial(t,x,u)\right|\ne0$.
It is known that for evolution equations we have the restrictions, $T_x=T_u=0$, on the general form of admissible transformations~\cite{Kingston&Sophocleous1998} and moreover for equations of the form~$u_t=F(t,x,u)u_{xx}+G(t,x,u,u_x)$ we necessarily have the condition $X_u=0$~\cite{Popovych&Ivanova2004}.
Therefore it is enough to consider a transformation, $\mathcal T$, of the form
\[
\tilde t=T(t), \quad \tilde x=X(t,x), \quad \tilde u=U(t,x,u),
\]
where $T_tX_xU_u\ne0$.
After we change the variables
in~\eqref{Eq_GenBurgers_tilde}, we obtain an equation in the variables without tildes.
It should be an identity on the manifold~$\mathcal L$ determined by~\eqref{Eq_GenBurgers} in
the second-order jet space~$J^2$ with the independent variables $(t,x)$ and the dependent variable~$u$.
To involve the constraint between variables of~$J^2$ on the manifold~$\mathcal L$,
we substitute the expression of~$u_t$ implied by equation~\eqref{Eq_GenBurgers}.
The splitting of this identity with respect to the derivatives $u_{xx}$ and $u_x$
results in the determining equations for the functions~$T$, $X$ and $U$
\begin{gather}\label{1}
U_{uu}=0,\qquad
\tilde g T_t -g X_x^2=0,\\\label{3}
X_t U_x-X_x U_t+\tilde g T_t{\left(\frac{U_x}{X_x}\right)\!}_x-\tilde a\tilde n T_t U^{\tilde n-1}U_x=0 \quad \mbox{\rm and} \\\label{4}
anu^{n-1}-\tilde a\tilde n\frac{T_t}{X_x}U^{\tilde n-1}+2\tilde g\frac{T_t}{X_x^2}\frac{U_{xu}}{U_u}-\tilde g T_t\frac{X_{xx}}{X_x^3}+\frac{X_t}{X_x}=0.
\end{gather}
Equations~\eqref{1}  imply that
 \[U=\eta^1(t,x)u+\eta^0(t,x),\quad X=\varphi(t)x+\psi(t) \quad\mbox{and}\quad\tilde g=\frac{\varphi^2}{T_t}g.\]
 Here the functions $\eta^i(t,x)$, $i=1,2$, $\varphi(t)$, and $\psi(t)$ are arbitrary smooth functions of their arguments and $\eta^1\varphi\neq0.$

When we use the differential consequences of the fourth equation with respect to $u$, we get that the arbitrary element $n$ is invariant under the action of a point transformation, i.e.,
$\tilde n =n.$
Also we obtain that $\eta^0=0$ for any $n\neq2.$
After we substitute the expressions for $U$, $X$ and~$\tilde g$ into the third and the fourth determining equations, we can split them with respect to $u$.
Further consideration varies depending upon  whether  $n\neq2$ or $n=2$.

\bigskip

{\rm\bf I}. If $n\neq2$, then splitting of~\eqref{3} and~\eqref{4} results in
the equations
\begin{gather*}
\eta^1_x=0, \quad \eta^1(\varphi_tx+\psi_t)+2\varphi g\eta^1_x=0,\\
\varphi\eta^1_t=\eta^1_x(\varphi_tx+\psi_t)+\varphi g\eta^1_{xx}, \quad \mbox{\rm and} \\
\varphi a \eta^1=\tilde a (\eta^1)^nT_t.
\end{gather*}
The general solution of this system is given by
\[T=\delta_1t+\delta_2,\quad\varphi=\delta_3,\quad\psi=\delta_4,\quad \eta^1=\delta_5,\]
where $\delta_j,$ $j=1,\dots,5,$ are arbitrary constants with $\delta_1\delta_3\delta_5\neq0.$ Then $\tilde a=\dfrac{\delta_3}{\delta_1}\delta_5^{1-n}a $ and
$\tilde g=\dfrac{{\delta_3}^2}{\delta_1} g.$

The statement of Theorem~1 is proved.

\bigskip

{\rm\bf II}. If $n=2$, then splitting of equations~\eqref{3} and~\eqref{4} leads to the system
\begin{gather*}
\eta^1_x=0, \quad \eta^1(\varphi_tx+\psi_t)+2\varphi g\eta^1_x-2\tilde a T_t\eta^0\eta^1=0,\\
\varphi\eta^1_t=\eta^1_x(\varphi_tx+\psi_t)+\varphi g\eta^1_{xx}-2\tilde a T_t(\eta^0\eta^1)_x, \\
\varphi\eta^0_t=\eta^0_x(\varphi_tx+\psi_t)+\varphi g\eta^0_{xx}-2\tilde a T_t\eta^0\eta^0_x, \quad \mbox{\rm and}  \\
\varphi a =\tilde a \eta^1T_t.
\end{gather*}
From this system we initially obtain forms of $\eta^1 $ and  $\eta^0$ as
\[\eta^1=\sigma\frac{\varphi}{T_t},\quad\eta^0=\frac{\sigma}{2aT_t}(\phi_tx+\psi_t),\quad \tilde a=\frac{a}{\sigma},\]
where $\sigma$ is a nonzero constant. The remaining equations for the functions $T,$ $\varphi$ and $\psi$ are
\[\left(\frac{\varphi^2}{T_t}\right)_t=0,\quad\left(\frac{\varphi_t}{T_t}\right)_t=0,\quad\left(\frac{\psi_t}{T_t}\right)_t=0.\]
Their general solution can be written as
\[T=\frac{\alpha t+\beta}{\gamma t+\delta},\quad \varphi=\frac{\kappa}{\gamma t+\delta},\quad \psi=\frac{\mu_1 t+\mu_0}{\gamma t+\delta},\]
where $\alpha$, $ \beta$, $\gamma$, $\delta$, $\kappa$, $\mu_1$ and $ \mu_0$ are constants  defined up to a nonzero multiplier,
 $\alpha\delta-\beta\gamma\neq0$ and $\kappa\not=0$.  When we substitute the functions $T,$ $\varphi$ and $\psi$ into the formulae for $\eta^1 $ and  $\eta^0$,
 we obtain exactly the statement of Theorem~2. The group $\hat G^{\sim}_2$ is called generalized since  transformation component for $u$ depends upon an arbitrary element, $a$, of the class.\footnote{Note that, if $a=1/n$, the equivalence group~$\hat G^\sim_2$ was found previously in~\cite{Kingston&Sophocleous1991} (see also~\cite{Pocheketa&Popovych2012,Pocheketa2013}) in course of the study of  form-preserving (admissible)
transformations of the class of generalized Burgers equations,
$u_t+uu_x+f(t,x)u_{xx}=0.$}

We have found that all admissible transformations in class~\eqref{Eq_GenBurgers} are exhausted by those presented in Theorems 1 and 2. Therefore Theorems~1--3 are proven.
\end{proof}

\section{Lie symmetries}

We perform the group classification of class~\eqref{Eq_GenBurgers} within the framework of
the classical Lie approach~\cite{Olver1986,Ovsiannikov1982}.  It is convenient to perform the group classification for class~\eqref{Eq_GenBurgers} up to~$G^\sim$-equivalence and for its subclass~\eqref{Eq_GenBurgers_n2}
up to~$\hat G^\sim_2$-equivalence.

We search for operators of the form $\Gamma=\tau(t,x,u)\partial_t+\xi(t,x,u)\partial_x+\theta(t,x,u)\partial_u$
which generate one-parameter groups of point-symmetry transformations of an equation from class~\eqref{Eq_GenBurgers}.
Any such vector field,~$\Gamma$, satisfies the infinitesimal invariance criterion, i.e.,
the action of the second prolongation,~$\Gamma^{(2)}$, of the operator~$\Gamma$ on equation~\eqref{Eq_GenBurgers}
results in the conditions being an identity for all solutions of this equation. Namely, we require that
\begin{equation}\label{c1a}
\Gamma^{(2)}\{u_t+anu^{n-1}u_x-g(t)u_{xx}\}=0
\end{equation}
identically, modulo equation~\eqref{Eq_GenBurgers}.

The criterion of infinitesimal invariance implies that
\[
\tau=\tau(t),\quad
\xi=\xi(t,x), \quad
\theta=\theta^1(t,x)u+\theta^0(t,x),
\]
where $\tau$, $\xi$, $\theta^1$ and $\theta^0$ are arbitrary smooth functions of their arguments.
The remaining determining equations have the form
\begin{gather}\label{deteq1}
2g\xi_x=(g\tau)_t,\\\label{deteq2}
an\theta^1_xu^{n+1}+an\theta^0_xu^{n}+(\theta^1_t-g\theta^1_{xx})u^2+(\theta^0_t-g\theta^0_{xx})u=0,\\\label{deteq3}
an(\tau_t-\xi_x+(n-1)\theta^1)u^{n+1}+an(n-1)\theta^0u^n+(g\xi_{xx}-2g\theta^1_x-\xi_t)u^2=0.
\end{gather}
It is easy to see from~\eqref{deteq1} that $\xi_{xx}=0.$
The second and the third equations can be split with respect to different powers of $u$. Special cases of splitting arise if $n=0,1,2$.
If $n=0$ or $n=1$, equations~\eqref{Eq_GenBurgers} are linear and are excluded from consideration (Lie symmetries of second-order linear differential equations in two dimensions were studied over the century ago by Lie~\cite{Lie1881}). Therefore we investigate two cases, $n\neq2$ and $n=2$, separately.

{\bf I.} If $n\neq2$, then from~\eqref{deteq2} and~\eqref{deteq3} we immediately obtain that $\theta^1=c_0,$ where $c_0$ is an arbitrary constant, and $\theta^0=0$. Solving the rest of the equations derived from~\eqref{deteq3} after splitting we finally get the forms of $\tau$, $\xi$ and $\theta$,
\[\tau=c_1t+c_2,\quad \xi=(c_1+(n-1)c_0)x+c_3,\quad \theta=c_0u,\]
where $c_i,$ $i=1,\dots,3,$ are arbitrary constants.
Then~\eqref{deteq1} provides
the classifying equation on $g$,
\begin{equation}\label{c1}
(c_1t+c_2)g_t=(c_1+2(n-1)c_0)g.
\end{equation}
Further consideration is performed using {\it the method of furcate splitting} suggested in~\cite{Popovych&Ivanova2004}.
Any operator $\Gamma$ from the maximal Lie invariance algebra $A^{\max}$ equation~\eqref{c1}
gives some equations on~$g$ of the general form
\begin{equation*}
(p\,t+q)g_t=sg,
\end{equation*}
where $p,$ $q,$ and $s$ are constants. In general for all operators from $A^{\max}$
the number $k$ of such independent equations is no greater than 2 otherwise
they form an incompatible system on $g$. There exist three
non-equivalent cases for the value of $k$ given by $k=0,$ $k=1$ and $k=2$.

If $k=0,$ then~\eqref{c1} is identically zero and $c_1=c_2=c_0=0$.
So, if $g$ is arbitrary, we obtain that the kernel of maximal Lie invariance algebras of equations from~\eqref{Eq_GenBurgers} is the
one-dimensional algebra $\langle\partial_x\rangle$.

If~$k=1,$ then $g\in\{\varepsilon e^t,\varepsilon\, t^{\rho}\}\!\!\mod G^{\sim}$, where $\varepsilon=\pm1$ and $\rho\ne0$.
In the exponential case, $g=\varepsilon e^t$, we have
\[\Gamma=2(n-1)c_0\partial_t+((n-1)c_0x+c_3)\partial_x+c_0u\partial_u.\]
If $g=\varepsilon t^\rho$ and $\rho\neq0,$ then
\[\Gamma=c_1t\partial_t+\left(\tfrac12(\rho+1)c_1x+c_3\right)\partial_x+ \tfrac12\displaystyle{\tfrac{\rho-1}{n-1}}c_1u\partial_u.\]
In both these cases the maximal Lie-invariance algebras are two-dimensional with basis operators presented in Cases 2 and 3 of Table~1.

If $k=2$, $g=1\bmod G^\sim.$ The infinitesimal operator takes the form
\[\Gamma=(c_1t+c_2)\partial_t+\left(\tfrac12c_1x+c_3\right)\partial_x- \displaystyle{\tfrac1{2(n-1)}}c_1u\partial_u.\]
So, if $g$ is a constant, the maximal Lie invariance algebra of~\eqref{Eq_GenBurgers} with $n\neq2$ is three-dimensional
spanned by operators presented in Case 4 of Table~1.
Therefore we have proven the following statement.

\begin{theorem}
The kernel of the maximal Lie invariance algebras of equations from class~\eqref{Eq_GenBurgers} with $n\neq2$
coincides with the one-dimensional algebra $\langle\partial_x\rangle$.
All possible $G^\sim$-non-equivalent cases of extension of the maximal Lie invariance algebras are exhausted
by the cases 2--4 of Table~1.
\end{theorem}
\begin{table}[t!]\small\begin{center}\refstepcounter{table}\renewcommand{\arraystretch}{1.75}
\textbf{Table~\thetable.} Group classification of the class~$u_t+a(u^n)_x=g(t)u_{xx}$, $n\neq0,1$.
\\[2ex]
\begin{tabular}{|c|c|c|l|}\hline
no. &\hfil $n$& $g$ & \hfil Basis operators of $A^{\rm max}$ \\
\hline
1 & $\neq2$&$\forall$ & $\partial_x$
\\
\hline
2 & $\neq2$&$\varepsilon t^\rho$ & $\partial_x,\quad
2t\partial_t+(\rho+1)x\partial_x+ \displaystyle{\frac{\rho-1}{n-1}}u\partial_u$
\\
\hline
3  & $\neq2$& $\varepsilon e^t$ & $\partial_x,\quad 2\partial_t+x\partial_x+\frac1{n-1}u\partial_u$
\\
\hline
4 & $\neq2$ & $1$ & $\partial_x,\quad\partial_t,\quad
2t\partial_t+x\partial_x-\frac1{n-1}u\partial_u$
\\
\hline
5 & $2$&$\forall$ & $\partial_x,\quad t\partial_x+\partial_u$
\\
\hline
6 & $2$&$\varepsilon t^\rho$ & $\partial_x,\quad t\partial_x+\partial_u,\quad
2t\partial_t+(\rho+1)x\partial_x+(\rho-1)u\partial_u$
\\
\hline
7  & $2$& $\varepsilon e^t$ & $\partial_x,\quad t\partial_x+\partial_u,\quad 2\partial_t+x\partial_x+u\partial_u$
\\
\hline
8  & $2$& $\varepsilon e^{2\rho\arctan t}$ & $\partial_x,\quad t\partial_x+\partial_u,\quad(t^2+1)\partial_t+(t+\rho)x\partial_x+(x+(\rho-t)u)\partial_u$
\\
\hline
9 & $2$& $1$ & $\partial_x,\quad t\partial_x+\partial_u,\quad\partial_t,\quad 2t\partial_t+x\partial_x-u\partial_u, \quad t^2\partial_t+tx\partial_x+(x-tu)\partial_u$
\\
\hline
\end{tabular}
\end{center}
Here $\varepsilon=\pm1\bmod G^\sim$   and $\rho$ is a nonzero constant. In all cases $a=1/n\bmod  G^\sim$.
In Case~6 we can set, $\bmod\hat G^\sim_{2}$, either $\rho>0$ or $\rho<0$.
\end{table}
{\bf II.} If $n=2$, then splitting of~\eqref{deteq2} and~\eqref{deteq3} results in the system
\begin{gather*}
\theta^1_x=0,\quad 2a\theta^0_x+\theta^1_t=0,\quad \theta^0_t-g\theta^0_{xx}=0,\\
\tau_t-\xi_x+\theta^1=0,\quad 2a\theta^0-\xi_t=0.
\end{gather*}
The general solution of this system is
\begin{gather*}\tau=c_2t^2+c_1t+c_0,\quad \xi=\left(c_2t+\frac12c_1+c_5\right)x+2ac_3t+c_4,\\
\theta^1=-c_2t-\frac12c_1+c_5,\quad \theta^0=\frac1{2a}c_2x+c_3,
\end{gather*}
where $c_i$, $i=0,\dots,5,$ are arbitrary constants.
The classifying equation~\eqref{deteq1} takes the form
\begin{equation}\label{c2}
(c_2t^2+c_1t+c_0)g_t=2c_5\,g.
\end{equation}
For any operator $\Gamma$ from the  maximal Lie invariance algebra $A^{\max}$ equation~\eqref{c2}
gives some equations for~$g$ of the general form
\begin{equation}\label{Eqn2clas}
(p\,t^2+q\,t+r)g_t=s\,g,
\end{equation}
where $p,$ $q,$ $r,$ and $s$ are constants.  As in the previous case
the number $k$ of such independent equations is no greater than 2 otherwise
they form an incompatible system on $g$. So three
non-equivalent cases for the value of $k$ should be considered, namely $k=0,$ $k=1$ and $k=2.$

If $k=0,$ then~\eqref{c2} is identically zero and $c_0=c_1=c_2=c_5=0$.
So, if $g(t)$ is arbitrary, we obtain that the kernel of the maximal Lie invariance algebras of equations from~\eqref{Eq_GenBurgers_n2} is the
two-dimensional algebra $\langle\partial_x,\,2at\partial_x+\partial_u\rangle$.

If $k=1$, the following statement is true.
\begin{lemma}\label{LemmaOntransOfCoeffsOfClassifyingSystem2}
Up to $\hat G^{\sim}_{2}$-equivalence
the parameter quadruple~$(p,q,r,s)$ can be assumed to belong to the set
\[
\{(0,1,0,\bar s),\ (0,0,1,1),\ (1,0,1,s')\},
\]
where $\bar s$, $s'$ are nonzero constants, $\bar s>0$.
\end{lemma}
\begin{proof}
Combined with multiplication by a nonzero constant,
each transformation from the equivalence group~$\hat G^{\sim}_{2}$ can be  extended to the coefficient quadruple
of equation~\eqref{Eqn2clas} as
\begin{gather*}
\begin{array}{l}
\tilde p=\nu(p\delta^2-q\gamma\delta+r\gamma^2),\quad
\tilde q=\nu(-2p\beta\delta+q(\alpha\delta+\beta\gamma)-2r\alpha\gamma),
\\[1ex]
\tilde r=\nu(p\beta^2-q\alpha\beta+r\alpha^2),
\quad\tilde s=\nu s\Delta.
\end{array}
\end{gather*}
Here $\Delta=\beta\gamma-\alpha\delta$ and $\nu$ is an arbitrary nonzero constant.

There are only three $\hat G^\sim_2$-non-equivalent values of the triple $(p,q,r)$
depending upon the sign of $D=q^2-4pr$,
\begin{gather*}
(0,1,0)\quad\mbox{if}\quad D>0, \quad
(0,0,1)\quad\mbox{if}\quad  D=0 \quad \mbox{\rm and} \quad
(1,0,1)\quad\mbox{if}\quad  D<0.
\end{gather*}
Indeed, if $D>0$, then there exist two linearly independent pairs
$(\delta,\gamma)$ and $(\alpha,\beta)$
such that $p\delta^2-q\gamma\delta+r\gamma^2=0$ and $p\beta^2-q\alpha\beta+r\alpha^2=0$.
For these values of the constants, $\alpha$, $\beta$, $\gamma$ and~$\delta$, we have $\tilde p=\tilde r=0$.
The coefficient $\tilde q$ is necessarily nonzero and it
can be scaled to~$1$ using multiplication by an appropriate value of~$\nu$. Certain freedom in varying group parameters is preserved even after fixing  of the form
of the triple $(p,q,r)$ and $(\tilde p,\tilde q,\tilde r)$.
This allows us to set a constraint  for the coefficient~$\tilde s$. Thus the transformation $\tilde t=1/t$ alternates the sign of $\tilde s$ so that it can be assumed to be positive one.

In the case $D=0$ we choose values of $\alpha$, $\beta$, $\gamma$ and $\delta$
for which $p\delta^2-q\gamma\delta+r\gamma^2$ and
the pair $(\delta,\gamma)$ is not proportional to the pair $(\alpha,\beta)$.
Then we obtain that $\tilde p=0$ and
$\tilde q=\nu\beta(q\gamma-2p\delta)+\nu\alpha(\delta q-2r\gamma)=0$.
An appropriate choice of the pair $(\alpha,\beta)$ allows us to set $\tilde r=1$. Then the residual constant $\tilde s$ can be scaled to one by choice of $\nu.$

If $D<0$, we have $pr\ne0$ and can set $p>0$.
We always can set $\tilde p=\tilde r=1$ and $\tilde q=0$, e.g., this gauge can be made by the transformation
$\tilde t=(2pt+q-\sqrt{4pr-q^2})/(2pt+q+\sqrt{4pr-q^2}).$ In this case the constant $\tilde s$ cannot  be scaled.
\end{proof}
Therefore, up to $\hat G^\sim_2$-equivalence, we have three cases of $g$ which provide extension of the Lie symmetry algebra by one basis operator.
These are Cases 6--8 of Table~1.

If $k=2$ and $g=1\bmod \hat G^\sim_2$,  we get a five-dimensional Lie symmetry algebra which is $\mathfrak{sl}(2,\mathbb{R})\lsemioplus 2A_1$ (Case 9 of Table~1).

The arbitrary element  $a$ does not affect the results of the group classification problem and can be scaled to any fixed nonzero constant value by the transformations
from the usual equivalence group~$G^\sim$. It is convenient to perform the gauge $a=1/n.$

Therefore we have proven the following statement.
\begin{theorem}
The kernel of the maximal Lie invariance algebras of equations from class~\eqref{Eq_GenBurgers_n2}
coincides with the two-dimensional Abelian algebra $\langle\partial_x,\,2at\partial_x+\partial_u\rangle$.
All possible $\hat G^\sim_2$-non-equivalent cases of extension of the maximal Lie invariance algebras are exhausted
by the cases 6--9 of Table~1.
\end{theorem}

\section{Solution of a boundary-value problem using Lie symmetries}
We consider the class of BVPs
\begin{eqnarray}\label{BV_Eq_GenBurgers}
&&u_t+a(u^n)_x=g(t)u_{xx},~~~  x\in [0,+\infty),~t>0, \nonumber\\
&&\lim_{t\rightarrow+\infty} u(t,x)=0,~~~x \in (0,+\infty),\\
&&u(t,0)=q(t),~~~t>0, \nonumber\\
&&\lim_{x\rightarrow+\infty}u(t,x)=0,~~~t>0,\nonumber
\end{eqnarray}
where $a$ is a nonzero constant, $g$ and $q$ are arbitrary smooth nonvanishing functions and $n\neq0,1$ and search those for which
the ``direct'' approach suggested by Bluman~\cite{Bluman&Cole1969,Bluman&Anco2002} is applicable.

We have derived the Lie symmetries for the variable coefficient equation  (\ref{Eq_GenBurgers}) and now we examine which of these symmetries leave the initial and boundary conditions of the problem~(\ref{BV_Eq_GenBurgers}) invariant. The procedure starts by assuming a general symmetry of the form
\begin{equation}\label{general_symmetry}
\Gamma =\sum_{i=1}^m\alpha_i\Gamma_i,
\end{equation}
where $m$ is the number of basis operators of maximal Lie symmetry algebra of a given PDE and  $\alpha_i,~i = 1,\dots,m$, are constants to be determined.

Lie symmetries for equation (\ref{Eq_GenBurgers}) appear in Table 1. In Case 2, for which $g(t)=\varepsilon t^{\rho}$, the generator (\ref{general_symmetry}) takes the form
\[
\Gamma=\alpha_1\partial_x+\alpha_2\Big(2t\partial_t+(\rho+1)x\partial_x+\frac{\rho-1}{n-1}u\partial_u\Big).
\]
Application of $\Gamma$ to the first boundary condition which is written as
$
x=0$ and $u(t,0)=q(t)
$
gives
\[
\alpha_1=0\quad\mbox{ and }\quad\alpha_2\left(-2t\frac{dq}{dt}+\frac{\rho-1}{n-1} q \right)=0.
\]
For nonzero $\alpha_2$ we have
\[
q(t)=\gamma t^{\frac{\rho-1}{2n-2}},
\]
where $\gamma>0$ is a constant.
It can be shown that the symmetry $\Gamma$  with $\alpha_1=0$ leaves invariant the other boundary conditions. Hence the admitted Lie symmetry can be used to reduce BVP~(\ref{BV_Eq_GenBurgers}) to a problem with the governing equation being an ordinary differential equation. In fact the Lie symmetry $2t\partial_t+(\rho+1)x\partial_x+\left((\rho-1)/(n-1)\right)u\partial_u$ produces the transformation 
\begin{equation}\label{ansatz}
u=t^{\frac{\rho-1}{2n-2}}\phi(\eta),\quad\mbox{where}\quad\eta=xt^{-\frac{\rho+1}2},
\end{equation}
that reduces (\ref{BV_Eq_GenBurgers}) into the BVP for ODE
\begin{eqnarray}\label{redeq1}
&&2\varepsilon\phi''+(\rho+1)\eta\phi'-2a(\phi^{n})'-\frac{\rho-1}{n-1}\phi=0,\quad \eta\in[0,+\infty), \\\label{redeq2}
&&\phi(0)=\gamma, \\\label{redeq3}
&&\lim_{\eta\rightarrow+\infty} \phi(\eta)=0.
\end{eqnarray}

Let $\rho=(2-n)/n$.  Then~\eqref{redeq1} takes the form
$\varepsilon\phi''+(\eta\phi'+\phi)/n-a(\phi^{n})'=0$
and can be integrated once to give $\varepsilon\phi'+\eta\phi/n-a\phi^{n}+c=0,$
where $c$ is an integration constant.
When we set $c=0$, this equation becomes the Bernoulli equation that is linearizable by the substitution $\phi^{1-n}=z$ to the form \[\frac{\varepsilon}{1-n}z'+\frac1n\eta z-a=0.\]  The general solution of this equation is
\[z=e^{-\frac{1-n}{2n\varepsilon}\eta^2}\left(C+ \frac{a(1-n)}{\varepsilon}\int^{\eta}_0e^{\frac{1-n}{2n\varepsilon}\theta^2}{\rm d}\theta\right),\]
where  $C$ is an arbitrary constant.
If $\varepsilon n(n-1)>0$, the solution can be written in terms of the error function as
\[z=e^{\frac{\eta^2}{\sigma^2}}\left(C+ \frac{a(1-n)\sqrt{\pi}}{2\varepsilon\sigma}\operatorname{erf}(\sigma\eta)\right),\quad\mbox{where}\quad \sigma=\sqrt{\frac{n-1}{2\varepsilon n}},\quad {\rm erf}(\theta)=\frac{2}{\sqrt{\pi}}\int_0^{\theta}\!e^{-s^2}\!{\rm d}s.\]
Therefore a particular solution of the second-order ODE on the function $\phi$ is
\begin{equation}\label{sol_phi}
\phi=\begin{cases}e^{-\frac{1}{2\varepsilon n}\eta^2}\left(C+ \frac{a(1-n)}{\varepsilon}\int^{\eta}_0e^{\frac{1-n}{2n\varepsilon}\theta^2}{\rm d}\theta\right)^{\frac1{1-n}},\quad \mbox{if}\quad \varepsilon n(n-1)<0,\\
e^{-\frac{1}{2\varepsilon n}\eta^2}\left(C+ \frac{a(1-n)\sqrt{\pi}}{2\varepsilon\sigma}\operatorname{erf}(\sigma \eta)\right)^{\frac1{1-n}}, \quad\mbox{if}\quad \varepsilon n(n-1)>0,
\end{cases}
\end{equation}
where $\sigma=\sqrt{(n-1)/(2\varepsilon n)}.$
This is the solution of  BVP~\eqref{redeq1}--\eqref{redeq3} with $\rho=(2-n)/n$, when $C=\gamma^{1-n}$ and $\varepsilon n>0$. Its typical behaviour is shown on Figure~1.
\begin{figure}[t!]
\centering
\includegraphics[width=75mm]{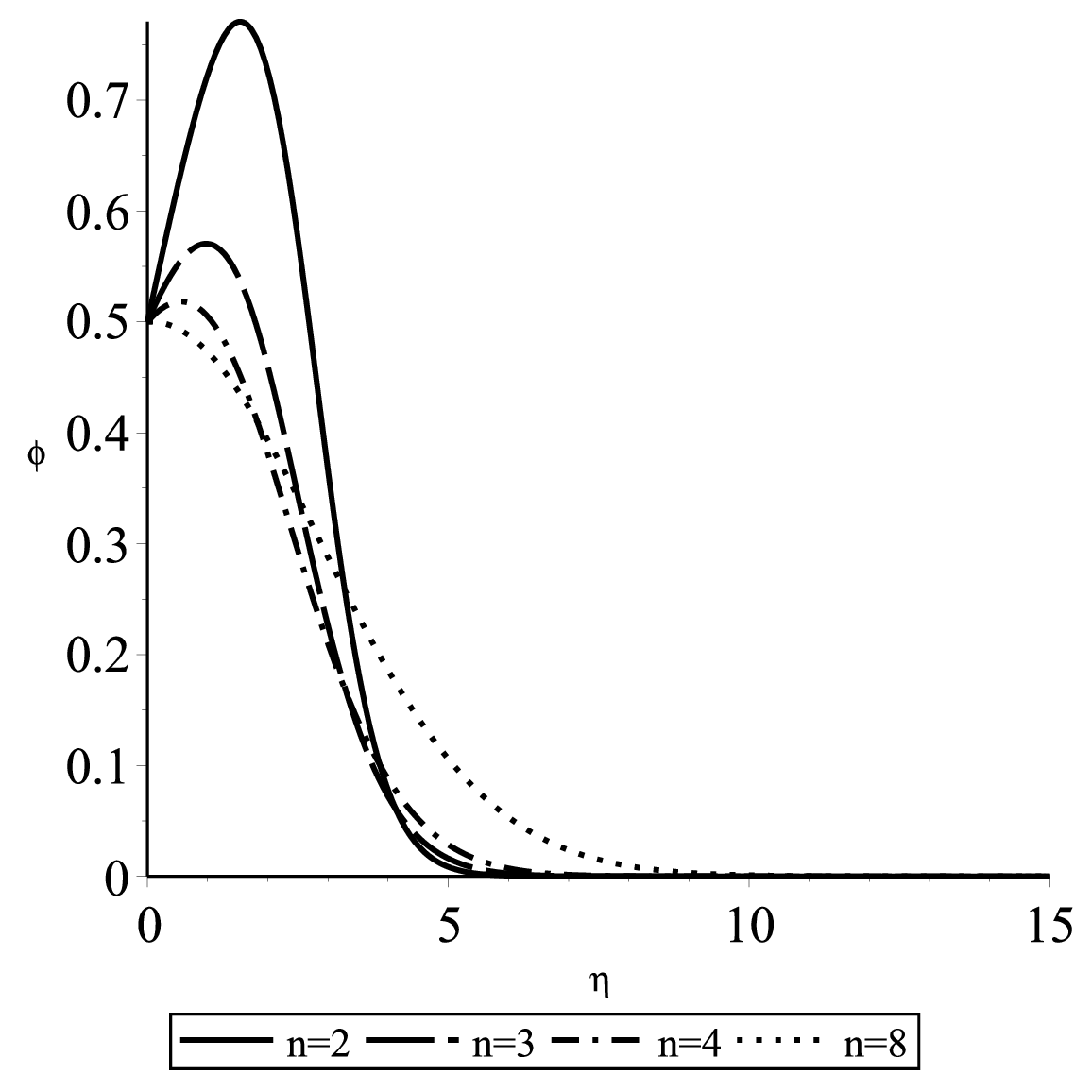}
\caption{ Solution~\eqref{sol_phi} for $\varepsilon=1$, $\gamma=0.5$, $a=1$ and various $n$.}
\end{figure}

Now we use~\eqref{ansatz} and obtain the solution of the following BVP
\begin{eqnarray}\label{BV_Eq_GenBurgers1}
&&u_t+a(u^n)_x=\varepsilon t^{\frac{2-n}n}u_{xx},~~~ x\in [0,+\infty),~t>0, \nonumber\\
&&\lim_{t\rightarrow+\infty} u(t,x)=0,~~~ x\in (0,+\infty),\\
&&u(t,0)=\gamma t^{-\frac1n},~~~t>0, \nonumber\\
&&\lim_{x\rightarrow+\infty}u(t,x)=0,~~~t>0.\nonumber
\end{eqnarray}
For $\varepsilon>0$ and $n>1$ the solution has the form
\begin{equation}\label{sol_u}
u=t^{-\frac1n}\exp\left[-\frac{1}{2\varepsilon n}x^2 t^{-\frac2n}\right]\left(\gamma^{1-n} +\frac{a(1-n)\sqrt{\pi}}{2\varepsilon\sigma}\operatorname{erf}(\sigma xt^{-\frac{1}n})\right)^{\frac1{1-n}},\quad \sigma=\sqrt{\frac{n-1}{2n\varepsilon}}.\end{equation}

Note that the solution satisfies BVP~\eqref{BV_Eq_GenBurgers1} for all values of positive $\gamma$ if $a<0$. If $a>0$, then the parameters have to satisfy the inequality  $\gamma^{1-n}>a(n-1)\sqrt{\pi}/(2\varepsilon\sigma).$
The typical behaviour of the latter solution is shown on Figures~2 and 3 for the values $n=3$ and $n=8$, respectively.

\begin{figure}[h!]
\begin{minipage}[t]{75mm}
\centering
\includegraphics[width=72mm]{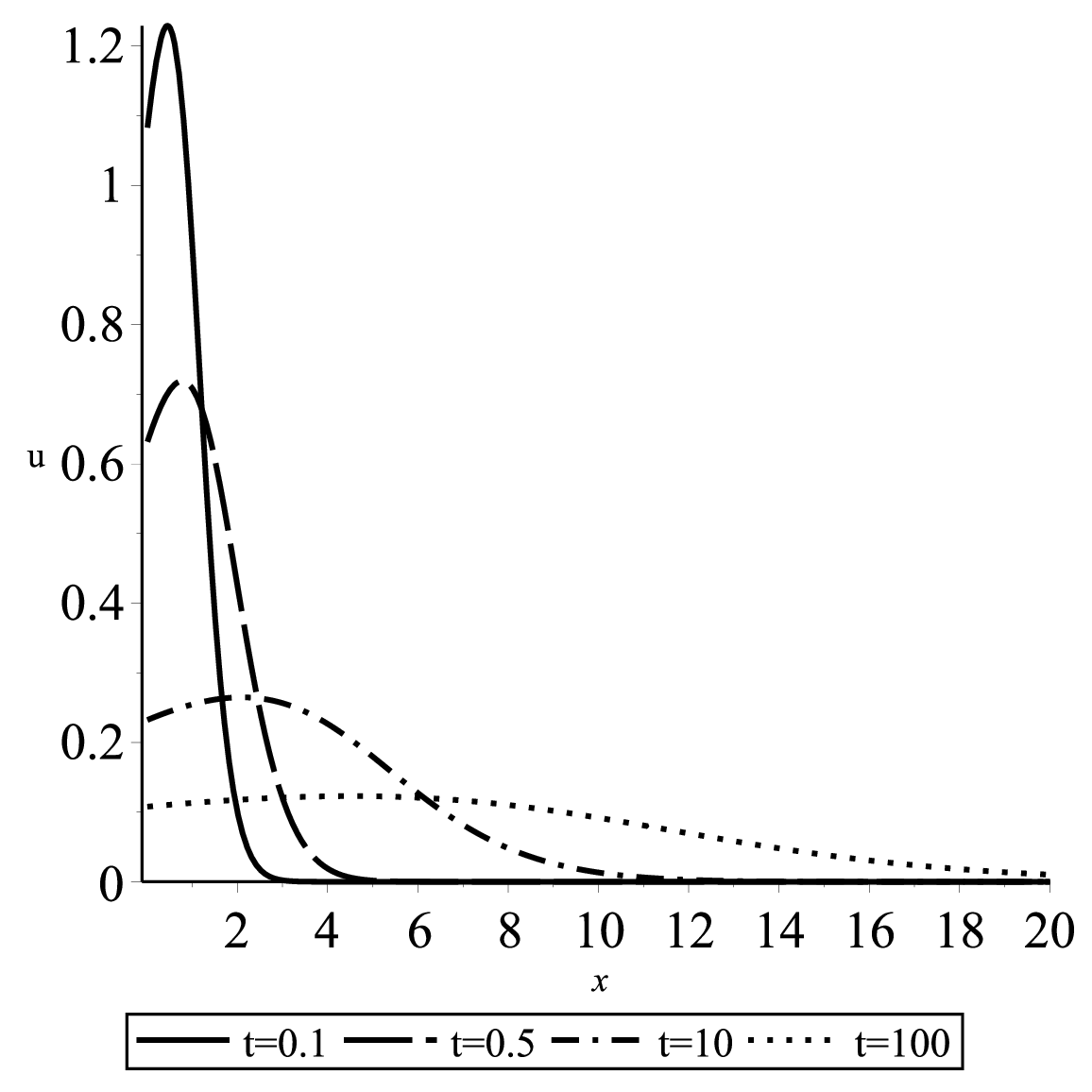}
\caption{Solution~\eqref{sol_u} for $\varepsilon=1$, $\gamma=0.5$, $a=1$ and $n=3$ (evolution in time).}
\end{minipage}
\quad
\begin{minipage}[t]{75mm}
\centering
\includegraphics[width=72mm]{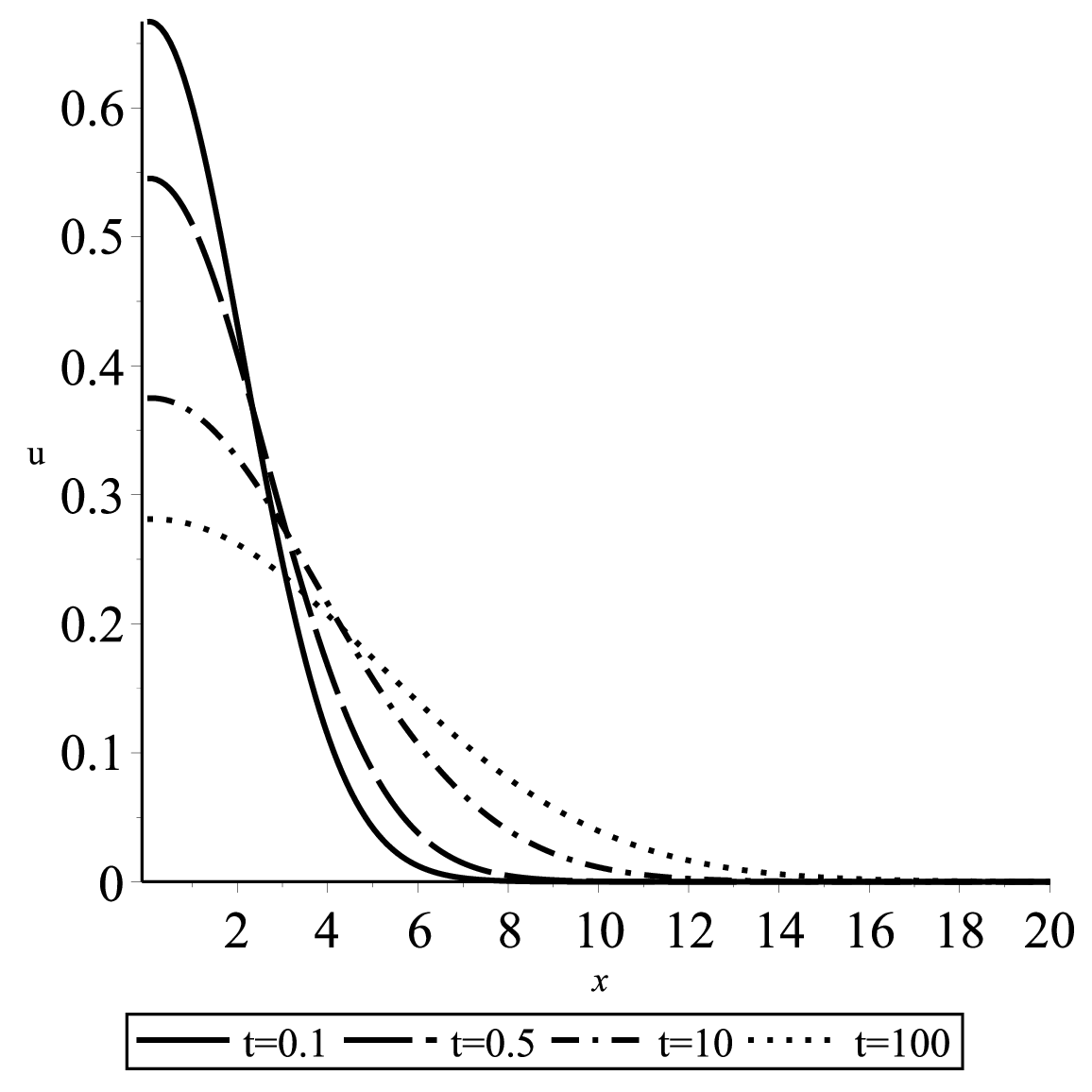}
\caption{Solution~\eqref{sol_u} for $\varepsilon=1$, $\gamma=0.5$, $a=1$ and $n=8$ (evolution in time).}
\end{minipage}
\end{figure}

The above procedure can be applied to the remaining cases that appear in Table 2. If we omit Case 9 which is the well-known Burgers equation,  and constant coefficient Case~4, only in Case 6 there exists a Lie symmetry that leaves the boundary and initial conditions invariant. However, the results for this case can be obtained from the above by setting $n=2$ and BVP~\eqref{BV_Eq_GenBurgers1} reduces to one with a constant coefficient governing equation.

\section*{Conclusion}

One performing research in the fields of engineering or physical sciences often encounters the problem of solving boundary-value problems (BVPs) for nonlinear partial differential equations.
Of course it is important to choose the method for solution which is easier to implement and which leads to more general results than others.
Some of the analytical methods are based on the usage of Lie symmetry groups.
In this paper we have applied  the classical ``direct'' technique involving Lie symmetries of PDEs~\cite{Bluman&Anco2002} to the class of BVPs for generalized Burgers equation with time-dependent viscosity coefficient and have solved the particular subcase. We have the opinion that the used approach is more straightforward  than the one suggested in~\cite{Moran&Gaggioli1969}. One more disadvantage of the latter technique is that it uses only
scalings and translations. Lie symmetry groups of some BVPs are wider and are not exhausted by scalings and translations only (see, e.g.,~\cite{Kovalenko}).
So, the ``direct'' approach is also more general. However, as we have seen, the present method is applicable only for specific forms of $g(t)$ and $q(t)$ in (14). In other words, the method has its limitations, but nevertheless is applied to nonlinear problems. Some recent examples of its successful usage can be found in~\cite{vane2014a,vane2014c}.

To take advantage of Lie symmetry method we firstly perform
the group classification  for the class of variable-coefficient  generalized Burgers equations~\eqref{Eq_GenBurgers} in the framework of modern group analysis. As a preliminary step we  investigate equivalence transformations within the class. It is shown that the equivalence group of the subclass of class~\eqref{Eq_GenBurgers} singled out by the condition $n=2$ is wider than the equivalence group of the whole class. Therefore the group classification list is presented in Table~1 up to $G^\sim$-equivalence for equations~\eqref{Eq_GenBurgers} with $n\neq2$ and up to $\hat G^\sim_2$-equivalence for those with $n=2$. The equivalence transformations have allowed us to write the classification list in a compact and a convenient form for further usage.
The  list of Lie symmetries for the class \eqref{Eq_GenBurgers} comes to complete the existing results that appear in the literature~\cite{Doyle&Englefield1990,wafo2004d} and presents all non-equivalent cases for which the algorithmic Lie reduction method can be applied.

\bigskip

{\bf Acknowledgements. }{\small
The authors thank the five referees for their constructive suggestions for the improvement of
this paper.
OV is grateful for the hospitality and financial support by the University of Cyprus and to Roman Popovych and Sergii Kovalenko for useful comments.  PGLL thanks the University of Cyprus for its kind hospitality and the University of KwaZulu-Natal and the National Research Foundation of South Africa
for their continued support.}

\end{document}